\documentclass[11pt]{amsart}

\usepackage{amssymb,latexsym}
\usepackage{amsmath}
\usepackage{amsfonts}
\usepackage{amsopn}
\usepackage{amsthm}
\usepackage{epsfig}
\usepackage{lscape}
\usepackage{multirow}
\usepackage{url}
\usepackage[pdftex,bookmarks=true]{hyperref}
\usepackage{pgf,tikz}
\usetikzlibrary{arrows}
\evensidemargin\oddsidemargin \topmargin-15mm

\newcounter{NN}
\newtheorem{remark}[NN]{Remark}

\newtheorem{theorem}[NN]{Theorem}

\newtheorem{lemma}[NN]{Lemma}
\newtheorem{definition}[NN]{Definition}

\def\Z{\mathbb{Z}}
\def\R{\mathbb{R}}
\def\Z{\mathbb{Z}}

\newcommand{\superscript}[1]{\ensuremath{^{\textrm{#1}}}}

\begin{document}
\title{Poisson brackets of mappings obtained as $(q,-p)$ reductions of lattice equations}
\author{D.T. Tran\superscript{1}, P.H. van der Kamp\superscript{2}, G.R.W. Quispel\superscript{2}}
\address{\superscript{1}School of Mathematics and Statistics,
University of New South Wales, Sydney NSW 2052, Australia}
\address{\superscript{2}Department of Mathematics and Statistics,
La Trobe University, Bundoora VIC 3086, Australia}
\begin{abstract}
In this paper, we present Poisson brackets of  certain classes of mappings obtained as general periodic reductions of integrable lattice equations. The Poisson brackets are derived using the so called Ostrogradsky transformation.
 \end{abstract}
 \maketitle

 \section{Introduction}
In this paper we derive Poisson structures (and sympletic structures)  for several families of discrete dynamical systems. In general, of course, the problem of finding a Poisson structure for a given mapping (if it has such a structure) is not algorithmic. However, suppose the system in question is the periodic reduction of a partial difference equation, e.g.
\begin{equation} \label{E:R1}
u_{l+1,m+1}=F(u_{l,m}, u_{l+1,m}, u_{l,m+1})
\end{equation}
with
\begin{equation} \label{E:R2}
u_{l+q,m-p}=u_{l,m},
\end{equation}
where $F$ is a given function, $\R^3\to\R$, and $u$ is the unknown function, $\Z^2\to\R$. In the simplest cases, $p=q$ resp. $p=1$, Poisson structures were first found by inspection, cf. \cite{Papageorgiou}, resp. \cite{TranInvo}.
Later on, for the case $p=-q$ (pre) symplectic structures for reductions of Adler-Bobenko-Suris (ABS) equations were derived using the existence of the so-called three-leg form \cite{ABS, Tran_thesis}. Using the three-leg forms of ABS equations, we can derive Lagrangians for these equations which help us to find symplectic structures  for the case where $p=\pm q$.
This method is called the Hamiltonian-Lagrangian approach and it  was introduced in \cite{Bruschietal, Ves}.
For the case $p=q$,
 using the discrete Legendre transformation we can bring  reduction of potential Korteweg-de Vries (pKdV) to canonical coordinates \cite{Capel}.  For  the case $p\neq \pm q$,    the discrete analogue of the so-called  Ostrogradsky transformation  will help us to  bring the map to canonical coordinates  \cite{Blaszak, Bruschietal, Ves}; therefore we can derive Poisson structures of the original maps. We note that in the
continuous context, this transformation has been using  to construct Hamiltonians from  higher derivative Lagrangians. The Hamiltonian-Lagrangian approach was used in \cite{Hone2013}\footnote{Another approach to obtain Poisson structures used in \cite{Hone2013} is based on cluster algebras.} to find the Poisson structure of $p=1$ reductions of the discrete KdV equation.

In this paper we use the Ostrogradsky transformation to find Poisson structures for periodic reductions of several partial difference equations, including the discrete KdV equation, the plus-KdV equation, the discrete Lotka-Volterra equation, and the discrete Liouville equation.
All the partial difference equations we discuss (except for one) are integrable, i.e. they possess a Lax representation, and this Lax representation can be used to derive Lax representations for the reductions, in their turn yielding integrals of the reductions, cf. \cite{Papageorgiou, QCPN1991, Tran}. As examples, we explicitly show that $p=2$, $q=3$ reductions of the integrable partial difference equations are Liouville integrable in their own right, i.e. they possess a sufficient number of functionally independent integrals in involution with respect to the derived Poisson structures.

\section{Setting}
In this section, we give some basic notions for the rest of the paper. In particular, we will present the so-called discrete Ostrogradsky transformation to obtain Poisson brackets from a Lagrangian.
 \subsection{Completely integrable map} %In this section, we consider $d$-dimensional  maps $\varphi: {\bf x}\mapsto {\bf x}^{'}$, where ${\bf x}=(x_1,x_2,\ldots, x_d)\in \R^d$.
Recall that a map between two  Poisson manifolds is called Poisson if it preserves  the Poisson brackets. In particular, a $d$-dimensional mapping
\begin{align}
\label{C4E:dDimeMap}
\varphi: \R^d&\rightarrow \R^d\notag\\
(x_1,x_2,\ldots,x_d)&\mapsto (x_1^{'},x_2^{'},\ldots,x_d^{'})
\end{align}
is called a Poisson map if
\begin{equation}
\label{C4E:Poimap}
\{x_i^{'},x_j^{'}\}=\{x_i,x_j\}\mid_{{\bf x}={\bf x}^{'}},
\end{equation}
where $\{,\}$ denotes the Poisson bracket in $\R^d$ (the same bracket in the domain and the co-domain).
One sufficient condition for the map~\eqref{C4E:dDimeMap} to be Poisson is the existence of a structure matrix, $\Omega$, such that
\begin{equation}
\label{E:Poisson1}
d\varphi. \Omega(x).  d\varphi^T=\Omega(x^{'}),
\end{equation}
where $d\varphi$ is the  Jacobian of the map~\eqref{C4E:dDimeMap}.
By structure matrix, we mean that $\Omega$ is a $d\times d$ skew-symmetric matrix that satisfies the Jacobi identity \cite{Olver00}
\begin{equation}
\label{C1E:JacoIden}
\sum_{l}\big(\Omega_{li}\frac{\partial}{\partial x_l}\Omega_{jk}+\Omega_{lj}\frac{\partial}{\partial x_l}\Omega_{ki}
+\Omega_{lk}\frac{\partial}{\partial x_l}\Omega_{ij}\big)=0,\ \forall i,j,k.
\end{equation}
The Poisson bracket between 2 smooth functions $f$ and $g$ is defined through this structure by the following formula
\begin{equation}
\label{C4E:PoiDefine}
\{f,g\}=\nabla f .  \Omega . (\nabla g)^{T},
\end{equation}
where $\nabla f$ is the gradient of $f$. In this case, $\Omega_{ij}=\{x_i,x_j\}$. We can see that the entries in the LHS and RHS of (\ref{E:Poisson1})
are  $\{x_i^{'},x_j^{'}\}$ and $\{x_i,x_j\}\mid_{{\bf x}={\bf x}^{'}}$, respectively.

If the Poisson structure is non-degenerate, i.e., $\Omega$ has  full rank,  one can define a closed two-form
\[
w=\sum_{i<j} w_{ij}dx_i\wedge dx_j
\]
through the
Poisson structure
$W=\Omega^{-1}$ and vice versa.\footnote{In the literature the symbol $J$ is sometimes used instead of $W$. We reserve $J$ to denote the
Jacobian.} Then the Poisson map  is called a symplectic map and preserves the closed two-form $W$.
 For the case where  the closed two-form is preserved by the map~\eqref{C4E:dDimeMap} but degenerate, we call the two-form
 pre-symplectic \cite{Byrnes1999}.

 A function $I({\bf x})$ is called an integral of $\varphi$ if $I({\bf x})=I({\bf x}^{'})$. A set of smooth  functions $\{I_1, I_2,\ldots, I_k\}$  on an open domain $D$ is called functionally independent on $D$ if its Jacobian matrix $dI$ has full rank on a set dense in $D$.
\begin{definition}
\label{d:Liouville}
The map $\varphi$ is completely integrable if it  preserves a Poisson bracket of rank $2r\leq d$ and it has $r+d-2r=d-r$ functionally independent integrals which are pairwise in involution.
 \end{definition}

 \subsection{$(q,-p)$ reductions  of lattice equations}
In this section, we introduce the $(q,-p)$-reduction ($\gcd(p,q)=1, q,p>0$) to obtain ordinary difference equations from lattice equations.

Given a lattice equation (see Figure~\ref{F:reduction_map})
\begin{equation}
\label{E:lattice_eq}
Q(u_{l,m}, u_{l+1,m}, u_{l,m+1},u_{l+1,m+1})=0,
\end{equation}
the $(q,-p)$ reduction can be described as follows. We introduce the  periodicity condition $u_{l,m}=u_{l+q,m-p}$ and travelling wave reduction
$v_n=u_{l,m}$, where $n=lp+mq$. For $0\leq n \leq p+q-1$, since $\gcd(p,q)=1$ one can find unique $(l,m), l\geq 0, m\leq 0$ such that $|m|$ is smallest and
$n=lp+mq$.
The lattice equation reduces to the following ordinary difference equation
\begin{equation}
\label{E:ordinary_eq}
Q(v_n,v_{n+p},v_{n+q},v_{n+p+q})=0.
\end{equation}
This equation gives us a $(p+q)$-dimensional map
\begin{equation}
\label{E:map}
(v_n,v_{n+1},\ldots,v_{p+q-1})\mapsto (v_{n+1},v_{n+2},\ldots, v_{n+p+q}),
\end{equation}
where we have assumed that one can solve $v_{n+p+q}$ uniquely from  equation~\eqref{E:ordinary_eq}.
 For example, the $(3,-2)$-reduction is depicted in Figure~\ref{F:reduction_map}.
 \definecolor{qqqqff}{rgb}{0.,0.,1.}
 \begin{figure}[h]
 \begin{center}
\begin{tikzpicture}[line cap=round,line join=round,>=triangle 45,x=1.5cm,y=1.5cm]
\clip(-3.4030033473878976,-3.3169037917447777) rectangle (5.356870154756367,-0.6225403676158588);
\draw (-2.,-2.)-- (-2.,-3.);
\draw (-2.,-3.)-- (-1.,-3.);
\draw (-1.,-2.)-- (-2.,-2.);
\draw (-1.,-2.)-- (-1.,-3.);
\draw (0.,-1.)-- (2.,-1.);
\draw (2.,-1.)-- (2.,-2.);
\draw (2.,-2.)-- (3.,-2.);
\draw (3.,-2.)-- (3.,-3.);
\draw (3.,-3.)-- (5.,-3.);
\draw [dash pattern=on 2pt off 2pt] (3.,-2.)-- (4.,-2.);
\draw [dash pattern=on 2pt off 2pt] (4.,-2.)-- (4.,-3.);
\begin{scriptsize}
\draw [fill=black] (-2.,-2.) circle (1.0pt);
\draw[color=black] (-2.,-2.) node[anchor=south]{$u_{l,m+1}$};
\draw [fill=black] (-1.,-2.) circle (1.0pt);
\draw[color=black] (-1.,-2.) node[anchor=south] {$u_{l+1,m+1}$};
\draw [fill=black] (-1.,-3.) circle (1.0pt);
\draw[color=black] (-1.,-3.)  node[anchor=north] {$u_{l+1,m}$};
\draw [fill=black] (-2.,-3.) circle (1.0pt);
\draw[color=black] (-2.,-3)  node[anchor=north]{$u_{l,m}$};
\draw [fill=black] (0.,-1.) circle (1.0pt);
\draw[color=black] (0.,-1.) node[anchor=south] {$v_0$};
\draw [fill=black] (1.,-1.) circle (1.0pt);
\draw[color=black] (1.,-1.) node[anchor=south] {$v_2$};
\draw [fill=black] (2.,-1.) circle (1.0pt);
\draw[color=black] (2,-1.)  node[anchor=south]{$v_4$};
\draw [fill=black] (2.,-2.) circle (1.0pt);
\draw[color=black] (2.15,-2)  node[anchor=south] {$v_1$};
\draw [fill=black] (3.,-2.) circle (1.0pt);
\draw[color=black] (3.,-2)  node[anchor=south] {$v_3$};
\draw [fill=black] (4.,-2.) circle (1.0pt);
\draw[color=black] (4.,-2)  node[anchor=south] {$v_5$};
\draw [fill=black] (3.,-3.) circle (1.0pt);
\draw[color=black] (3.,-3)  node[anchor=north] {$v_0$};
\draw [fill=black] (4.,-3.) circle (1.0pt);
\draw[color=black] (4.,-3)  node[anchor=north] {$v_2$};
\draw [fill=black] (5.,-3.) circle (1.0pt);
\draw[color=black] (5.,-3)  node[anchor=north] {$v_4$};
\end{scriptsize}
\end{tikzpicture}
\caption{ Quad equation and the $(3,-2)$ reduction.}\label{F:reduction_map}
\end{center}
\end{figure}
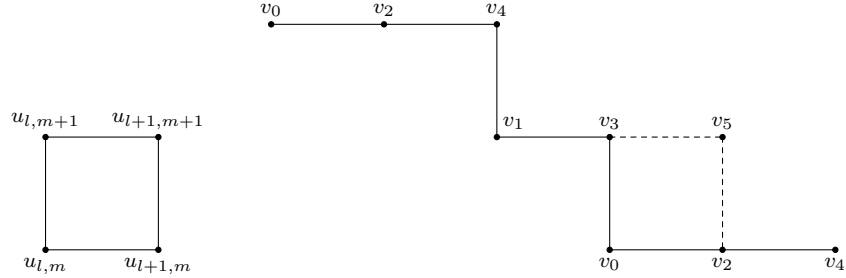

\subsection{Lagrangian equation and Ostrogradsky transformation}
%In this section,  present the method  given in \cite{Bruschietal}.
Given a mapping obtained from an Euler-Lagrange equation, Bruschi et al.  presented a transformation to rewrite this map in a canonical form \cite{Bruschietal}. We use this canonical form  to derive a Poisson structure
for mappings obtained from the Euler-Lagrange equation.

We start with the following symplectic  map without knowing its symplectic structure
\begin{equation}
\label{E:map1}
\phi : x(n)\mapsto x(n+1),
\end{equation}
 where $x(n),x(n+1)\in \R^{2N}$.
We rewrite the map in canonical form as follows
 \begin{equation}
 \label{C4E:canonical}
 \left(q(n),p(n)\right)\mapsto \left(q(n+1),p(n+1)\right),
 \end{equation}
  with $q_j(n+1)=g_j\left(q(n),p(n)\right)$ and $p_j(n+1)=f_j(q(n),p(n))$ and  such
that $\{f_i,f_j\}=\{g_i,g_j\}=0$ and $\{f_i,g_j\}=\delta_{ij}$ for $1\leq i,j\leq N$.

Suppose we have a Lagrangian $L(u(k),u(k+1),\ldots, u(k+N))$ and the discrete action functional
\begin{equation}
\label{C4E:DISfunctional}
I=\sum_{k\in\Z}L(u(k),u(k+1),\ldots, u(k+N)).
\end{equation}
The discrete Euler-Lagrange equation for the Lagrangian $L$ is
\[\frac{\delta I}{\delta u(n)}=0,\]
which means
\begin{multline}
\label{C4E:ELeq}
\frac{\partial L(u(n),u(n+1),\ldots,u(n+N))}{\partial u(n)}
+\frac{\partial L(u(n-1),u(n),\ldots,u(n+N-1))}{\partial u(n)}\notag \\
 + \ldots +\frac{\partial L(u(n-N),u(n-N+1),\ldots, u(n))}{\partial u(n)}
=0.
\end{multline}
We denote  $u^{(s)}=u(n+s)$ and $L_r=\frac{\partial L}{\partial u^{(r)}}$ and
introduce a shift operator $\mathcal{E}$, i.e,  $\mathcal{E}^j\left(u(n)\right)=u(n+j)$
for $j\in\Z$. The
Euler-Lagrange equation can be rewritten as follows
\[
\sum_{r=0}^N \mathcal{E}^{-r}L_r=0.
\]
\begin{definition}
\label{D:nomal}
A Lagrangian $L$ is called normal if $L_{0N}=\frac{\partial^2 L}{\partial u^{(0)}u^{(N)}}\neq 0$.
\end{definition}
Given a normal Lagrangian $L$,  we obtain a (non-linear) difference equation of order $2N$.
We introduce coordinates $q_i$ and  $p_i$, so that we can write the map derived from the Euler-Lagrange
equation in  canonical form. For example, we introduce a transformation
$$(u^{(-N)},\ldots,u^{(0)},\ldots, u^{(N-1)})\mapsto (q_1,q_2,\ldots,q_N,p_1,p_2,\ldots,p_N),$$
where
\begin{align}
q_i&=u^{(i-1)}=u(n+i-1), \label{C4E:OstraTran} \\
p_i&=\mathcal{E}^{-1}\sum_{k=0}^{N-i}\mathcal{E}^{-k}L_{k+i}, \label{C4E:OstraTran1}
\end{align}
for $i=1,2,\ldots, N$.
This transformation can be considered as an analog of  Ostrogradsky's transformation,  cf. \cite{DynaSys4}.
We have, using~\eqref{C4E:OstraTran},
\begin{equation}
\label{E:p1}
p_1=\sum_{k=1}^N\mathcal{E}^{-k}L_k=-L_0(q_1,q_2,\ldots,q_N, u^{(N)}),
\end{equation}
where we suppose that we can obtain
\[
u^{(N)}=\alpha(q,p_1).
\]
%The following result was presented in \cite{Bruschietal} without a proof.
We give a proof for this result.
\begin{lemma}
\label{L:mapcano}
The map~(\ref{C4E:canonical}) has the following canonical form
\begin{align}
q_i(n+1)&=g_i(q,p)=q_{i+1}(n), \ i=1,2,\ldots, N-1 , \label{E:qi}\\
q_N(n+1)&=g_N(q,p)=\alpha(q,p_1),\label{E:qN}\\
p_i(n+1)&=f_i(q,p)=p_{i+1}(n)+\widetilde{L}_i(q,p_1),\ i=1,2,\ldots , N-1, \label{E:pi}\\
p_N(n+1)&=f_N(q,p)=\widetilde{L}_N(q,p_1),\label{E:pN}
\end{align}
where we denote $\widetilde{L}(q,p_1)=L(q_1,q_2,\ldots, q_N, \alpha(q_1,q_2,\ldots, q_N,p_1))$.
\end{lemma}
\begin{proof}
It is easy to obtain~(\ref{E:qi}), (\ref{E:qN}) and~(\ref{E:pN}).
For $1\leq i\leq N-1$, we have
\begin{align*}
p_i(n+1)&=\sum_{k=0}^{N-i}\mathcal{E}^{-k}L_{k+i}=\sum_{k=1}^{N-i}\mathcal{E}^{-k}L_{k+i}+\widetilde{L}_{i}
=\sum_{k=0}^{N-i-1}\mathcal{E}^{-k-1}L_{k+i+1}+\widetilde{L}_i\\
&=p_{i+1}(n)+\widetilde{L}_i.
\end{align*}
Now we  prove that the mapping~\eqref{C4E:canonical} preserves the canonical
symplectic structure, i.e., we show that $\{f_i,f_j\}=\{g_i,g_j\}=0$ and $\{f_i,g_j\}=\delta_{ij}$.
For $1\leq i,j\leq N-1$, we have
 \begin{align*}
 \{g_i,g_j\}&=\{q_{i+1},q_{j+1}\}=0,\\
 \{g_i,g_N\}&=\{q_{i+1},\widetilde{L}_N(q,p_1)\}=-\frac{\partial \widetilde{L}_N(q,p_1)}{\partial p_{i+1}}=0,\\
 \{f_i,g_j\}&=\{p_{i+1},q_{j+1}\}+\{\widetilde{L}_i(q,p_1),q_{j+1}\}=
 \delta_{i+1j+1}+\frac{\partial \widetilde{L}_i(q,p_1)}{\partial p_{j+1}}=\delta_{ij},\\
 \{f_N,g_N\}&=\{\widetilde{L}_N(q,p_1),\alpha_1(q,p_1)\}=\frac{\partial \widetilde{L}_N(q,p_1)}{\partial p_1}
 \frac{\partial \alpha(q,p_1) }{\partial q_1}-\frac{\partial \widetilde{L}_N(q,p_1)}{\partial q_1}
 \frac{\partial \alpha(q,p_1) }{\partial p_1}\\
 &=\left(\L_{NN}\frac{\partial u^{(N)}}{\partial p_1}\frac{\partial u^{(N)}}{\partial q_1}-\big(L_{NN}\frac{\partial u^{(N)}}{\partial q_1}+L_{0N}\big)\frac{\partial u^{(N)}}{\partial p_1}\right)
 \mid_{u^{(0)}=q_1, \ldots, u^{(N-1)}=q_N,u^{(N)}=\alpha(q,p_1)}\\
 &=-L_{N0}\frac{\partial u^{(N)}}{\partial p_1}=\frac{\partial p_1}{\partial p_1}=1\\
 %\end{align*}
%\begin{align*}
\{f_i,f_j\}&=\{p_{i+1},p_{j+1}\}+\{p_{i+1},\widetilde{L}_j(q,p_1)\}+\{\widetilde{L}_i(q,p_1),p_{j+1}\}
 +\{\widetilde{L}_i(q,p_1),\widetilde{L}_j(q,p_1)\}\\
 &=\frac{\partial \widetilde{L}_j(q,p_1)}{\partial q_{i+1}}-\frac{\partial \widetilde{L}_i(q,p_1)}{\partial q_{j+1}}+
 \frac{\partial \widetilde{L}_i}{\partial p_1}\frac{\partial \widetilde{L}_j}{\partial q_1}-
 \frac{\partial \widetilde{L}_i}{\partial q_1}\frac{\partial \widetilde{L}_j}{\partial p_1}\\
 &=\Big(L_{ij}+L_{Nj}\frac{\partial u^{(N)}}{\partial q_{i+1}}-L_{ji}-L_{Ni}\frac{\partial u^{(N)}}{\partial q_{j+1}}+
 L_{Ni}\frac{\partial u^{(N)}}{\partial p_1}\big(L_{0j}+L_{Nj}\frac{\partial u^{(N)}}{\partial q_{1}}\big)\\
 & \quad -\big(L_{0i}+L_{Ni}\frac{\partial u^{(N)}}{\partial q_{1}}\big)L_{Nj}\frac{\partial u^{(N)}}{\partial p_1}\Big)\mid_{u^{(0)}=q_1, \ldots, u^{(N-1)}=q_N, u^{(N)}=\alpha(q,p_1)}\\
 &=\Big(L_{Nj}\frac{\partial u^{(N)}}{\partial q_{i+1}}-L_{Ni}\frac{\partial u^{(N)}}{\partial q_{j+1}}
 + L_{Ni}\frac{\partial u^{(N)}}{\partial p_1}L_{0j}
 \\
 &\quad -L_{0i}L_{Nj}\frac{\partial u^{(N)}}{\partial p_1}\Big)\mid_{u^{(0)}=q_1, \ldots, u^{(N-1)}=q_N,\ldots, u^{(N)}=\alpha(q,p_1)}\\
 \{f_i,f_N\}&=\{p_{i+1},\widetilde{L}_N(q,p_1)\}+\{\widetilde{L}_i(q,p_1),\widetilde{L}_N(q,p_1)\} \\
 &=\frac{\widetilde{L}_N(q,p_1)}{q_{i+1}}+\frac{\partial \widetilde{L}_i}{\partial p_1}\frac{\partial \widetilde{L}_N}{\partial q_1}-
 \frac{\partial \widetilde{L}_i}{\partial q_1}\frac{\partial \widetilde{L}_N}{\partial p_1}\\
 &=\Big(L_{iN}+L_{NN}\frac{\partial u^{(N)}}{\partial q_{i+1}}+L_{iN}L_{0N}\frac{\partial u^{(N)}}{\partial p_1}
 %\\
% &\quad
 -L_{i0}L_{NN}\frac{\partial u^{(N)}}{\partial p_1}\Big)\mid_{u^{(0)}=q_1,  u^{(N-1)}=q_N, u^{(N)}=\alpha(q,p_1)}.
 \end{align*}
 %\end{verbatim}
 We have $p_1=-L_0\mid_{u^{(0)}=q_1, \ldots, u^{(N-1)}=q_N, u^{(N)}=\alpha(q,p_1)}$, therefore
 \begin{align*}
 1&=-L_{0N}\frac{\partial u^{(N)}}{\partial p_1}\mid_{u^{(0)}=q_1, \ldots, u^{(N-1)}=q_N, u^{(N)}=\alpha(q,p_1)},\\
 0&=\left(-L_{0i}-L_{0N}\frac{\partial u^{(N)}}{\partial q_{i+1}}\right)\mid_{u^{(0)}=q_1, \ldots, u^{(N-1)}=q_N, u^{(N)}=\alpha(q,p_1)}.
 \end{align*}
 Thus, we get $\frac{\partial u^{(N)}}{\partial q_{i+1}}=L_{0i}\frac{\partial u^{(N)}}{\partial p_1}$. This implies that, $\{f_i,f_j\}=0$
 and $\{f_i,f_N\}=0$.
\end{proof}

 \section{Poisson brackets for $(q,-p)$ reductions of the KdV and the plus-KdV equations}
The KdV equation and a similar equation which we call the plus-KdV equation are given as follows
\begin{align}
u_{l+1,m}-u_{l,m+1}&=\frac{1}{u_{l,m}}-\frac{1}{u_{l+1,m+1}},\label{E:KdV}\\
u_{l+1,m}+u_{l,m+1}&=\frac{1}{u_{l,m}}+\frac{1}{u_{l+1,m+1}},\label{E:plusKdV}
\end{align}
respectively.
It  is known that KdV can be obtained from a double copy of potential KdV-equation or $H_1$ by introducing the new variable $v$ where $u_{l,m}=v_{l-1,m-1}-v_{l,m}$.
Similarly, for plus-KdV, we introduce variable $v$ such that $u_{l,m}=v_{l-1,m-1}+v_{l,m}$. In terms of variables $v$, equations~\eqref{E:KdV} and~\eqref{E:plusKdV}  become
\begin{align}
(v_{l+1,m}+v_{l-1,m}) - (v_{l,m-1}+v_{l,m+1})+\frac{1}{v_{l-1,m-1}-v_{l,m}}-\frac{1}{v_{l,m}-v_{l+1,m+1}}&=0,\label{E:KdV_new}\\
v_{l+1,m}+v_{l-1,m}+v_{l,m-1}+v_{l,m+1}-\frac{1}{v_{l-1,m-1}+v_{l,m}}-\frac{1}{v_{l,m}+v_{l+1,m+1}}&=0\label{E:plusKdV_new}.
\end{align}
These equations are Euler-Lagrange equations with the corresponding Lagrangians
\begin{align}
L^{1}_{l,m}&=v_{l,m}v_{l+1,m}-v_{l,m}v_{l,m+1}-\ln |v_{l,n}-v_{l+1,m+1}|, \label{E:Lag_KdV}\\
L^{2}_{l,m}&=v_{l,m}v_{l+1,m}+v_{l,m}v_{l,m+1}-\ln |v_{l,n}+v_{l+1,m+1}| \label{E:Lag_plus_KdV}.
\end{align}

The $(q,-p)$ reduction where $\gcd(p,q)=1$ and $p<q$ gives us the following Lagrangians
\begin{align}
L^{1}_n&=v_nv_{n+p}-v_{n}v_{n+q}-\ln |v_{n}-v_{n+p+q}|,\\
L^{2}_n&=v_nv_{n+p}+v_{n}v_{n+q}-\ln |v_{n}+v_{n+p+q}|.
\end{align}
The Euler-Lagrange equations derived from these Lagrangians are
\begin{align}
v_{n+p}+v_{n-p}-(v_{n+q}+v_{n-q})-\frac{1}{v_n-v_{n+p+q}}+\frac{1}{v_{n-p-q}-v_{n}}&=0,\label{E:double_KdV}\\
v_{n+p}+v_{n-p}+(v_{n+q}+v_{n-q})-\frac{1}{v_n+v_{n+p+q}}-\frac{1}{v_{n-p-q}+v_{n}}&=0,\label{E:double_plus_KdV}
\end{align}
which are the $(q,-p)$ reduction of  KdV and plus-KdV via reduction $u_{n}=v_{n-p-q}-v_n$.
In canonical coordinates, for a double copy of KdV we  have
\begin{align*}
Q_i&=v_{n+i-1},\ \mbox{for}\ 1\leq i\leq p+q,\\
P_1&=-(v_{n+p}-v_{n+q}-\frac{1}{v_n-v_{n+p+q}})=v_{n-p}-v_{n-q}+\frac{1}{v_{n-p-q}-v_n},\\
P_i&=v_{n+i-p-1}-v_{n+i-q-1}+\frac{1}{v_{n+i-p-q-1}-v_{n+i-1}},\ \mbox{for}\ 1<i\leq p,\\
P_i&=-v_{n+i-q-1}+\frac{1}{v_{n+i-p-q-1}-v_{n+i-1}},\ \mbox{for}\ p<i\leq q,\\
P_i&=\frac{1}{v_{n+i-p-q-1}-v_{n+i-1}},\ \mbox{for}\ q<i\leq p+q.
\end{align*}
For the plus-KdV equation, we have
\begin{align*}
Q_i&=v_{n+i-1},\ \mbox{for}\ 1\leq i\leq p+q,\\
P_1&=-(v_{n+p}+v_{n+q}-\frac{1}{v_n+v_{n+p+q}})=v_{n-p}+v_{n-q}-\frac{1}{v_{n-p-q}+v_n},\\
P_i&=v_{n+i-p-1}+v_{n+i-q-1}-\frac{1}{v_{n+i-p-q-1}+v_{n+i-1}},\ \mbox{for}\ 1<i\leq p,\\
P_i&=v_{n+i-q-1}-\frac{1}{v_{n+i-p-q-1}+v_{n+i-1}},\ \mbox{for}\ p<i\leq q,\\
P_i&=-\frac{1}{v_{n+i-p-q-1}+v_{n+i-1}},\ \mbox{for}\ q<i\leq p+q.
\end{align*}
Using the canonical Poisson bracket for these variables we obtain Poisson brackets for mappings obtained as the $(q,-p)$ reductions of
equations~\eqref{E:KdV} and~\eqref{E:plusKdV}. For both equations, we distinguish between two cases where $p=1$ and $p>1$.
\begin{theorem}
The KdV map admits the following  Poisson structure for $0\leq i <j\leq p+q-1$.
\begin{itemize}
\item If $p=1$,  we have
\begin{equation}
\label{E:poi_KdV}
\{u_{n+i},u_{n+j}\}=
\left\{
\begin{array}{ll}
(-1)^{j-i}\prod_{r=i}^j u_{n+i}^2, & 0<j-i<p+q-1,\\
\big(1+(-1)^{p+q-1}\prod_{r=i+1}^{j-1}u_{n+r}^2\big)u_{n+i}^2u_{n+j}^2, & j-i=p+q-1.
\end{array}
\right.
\end{equation}
\item If $p>1$, we have
\end{itemize}
\begin{equation}
\label{E:poi_KdV1}
\{u_{n+i},u_{n+j}\}=
\left\{
\begin{array}{ll}
u_{n+i}^2u_{n+j}^2& \mbox{if}\  j-i=q,\\
(-1)^{(j-i)/p}\prod_{k=0}^{(j-i)/p}u_{n+i+kp}^2,& \mbox{if} \ j-i\equiv 0 \mod p,\\
0&\mbox{otherwise}.
\end{array}
\right.
\end{equation}
\end{theorem}
\begin{proof}
We note that for the case $p=1$, the Poisson bracket was given in \cite{Hone2013}. We just need to prove this theorem for the case where $p>1$.

We first need to prove that the Poison bracket~\eqref{E:poi_KdV1} is preserved under the KdV map
\begin{equation}
\label{E:KdV_map}
\phi:\ (u_n,u_{n+1},\ldots, u_{n+p+q-1})\mapsto (u_{n+1},u_{n+2},\ldots, u_{n+p+q-1},u_{n+p+q}),
\end{equation}
where
\[
u_{n+p+q}=\frac{-u_n}{u_nu_{n+p}-u_nu_{n+q}-1}.
\]
By direct calculation, we have
\begin{equation}
\frac{\partial u_{n+p+q}}{\partial u_n}=-\frac{u_{n+p+q}^2}{u_n^2},\ \frac{\partial u_{n+p+q}}{\partial u_{n+p}}=u_{n+p+q}^2,\ \frac{\partial u_{n+p+q}}{\partial u_{n+q}}=-u_{n+p+q}^2.
\end{equation}
It is easy to see that for $0\leq i< j<p+q-1$, we have
\[
\{u_{n+i}^{'},u_{n+j}^{'}\}=\{u_{n+i+1},u_{n+j+1}\}=\{u_{n+i},u_{n+j}\}|_{{\bf{u}}=\phi({\bf{u}})}.
\]
For $0\leq i< j=p+q-1$, we have
\begin{align}
\label{E:Proof_Poisson}
\{u_{n+i}^{'},u_{n+p+q-1}^{'}\}&=\{u_{n+i+1},u_{n+p+q}\}\notag
\\
&=\{u_{n+i+1},u_n\}\frac{u_{n+p+q}^2}{u_n^2}+\{u_{n+i+1},u_{n+p}\}u_{n+p+q}^2-\{u_{n+i+1},u_{n+q}\}u_{n+p+q}^2.
\end{align}
By considering whether $p|(i+1)$,  one can see easily  that if  $i+1\neq q$ and $i+1\neq q$, then
\begin{equation}
\label{E:Proof_Poisson_1}
\{u_{n+i+1},u_n\}\frac{u_{n+p+q}^2}{u_n^2}+\{u_{n+i+1},u_{n+p}\}u_{n+p+q}^2=0.
\end{equation}
We consider $3$ cases.
\begin{itemize}
\item Case 1 where $p+q-i-1\neq q$ and $p\nmid (p+q-i-1)$. This means $i+1\neq p$ and $p\nmid q-i-1$.  Therefore, $\{u_{n+i+1},u_{n+q}\}=0$ and
$i+1\neq q$, $|i+1-p|\neq 0,\ \mbox{and}\ q$. This also  gives us \eqref{E:Proof_Poisson_1}. Thus,
$\{u_{n+i}^{'},u_{n+p+q-1}^{'}\}=0=\{u_{n+i},u_{n+p+q-1}\}|_{{\bf{u}}=\phi({\bf{u}})}.$
\item Case 2  where $p+q-i-1=q$, i.e. $i+1=p$. It implies that the second term and third term in \eqref{E:Proof_Poisson} vanish. Hence, we have
\begin{align*}
\{u_{n+p-1}^{'},u_{n+p+q-1}^{'}\}&
=\{u_{n+p},u_n\}\frac{u_{n+p+q}^2}{u_n^2}=u_{n+p}^2u_n^2\frac{u_{n+p+q}^2}{u_n^2}
=u_{n+p}^2u_{n+p+q}^2
\\
&
=u_{n+p-1}^2u_{n+p+q-1}^2|_{{\bf{u}}=\phi({\bf{u}})}=\{u_{n+p-1},u_{n+p+q-1}\}|_{{\bf{u}}=\phi({\bf{u}})}.
\end{align*}
\item Case 3 where $p|(p+q-i-1)$, i.e. $p|(q-i-1)$.
If $i+1=q$, the second and  third term in \eqref{E:Proof_Poisson} equals to $0$. Therefore, we have
\begin{align*}
\{u_{n+q-1}^{'},u_{n+p+q-1}^{'}\}&
=\{u_{n+q},u_n\}\frac{u_{n+p+q}^2}{u_n^2}
=-u_{n+q}^2u_{n+p+q}^2
\\
&
=-u_{n+q-1}^2u_{n+p+q-1}^2|_{{\bf{u}}=\phi({\bf{u}})}=\{u_{n+q-1},u_{n+p+q-1}\}|_{{\bf{u}}=\phi({\bf{u}})}.
\end{align*}
If $i+1\neq q$ and $i+1\neq p$, we have \eqref{E:Proof_Poisson_1}. Since $i+1-q<p$, we have $i+1-q<0$. Thus, \eqref{E:Proof_Poisson} becomes
\begin{align*}
\{u_{n+i}^{'},u_{n+p+q-1}^{'}\}&
=\{u_{n+i+1},u_{n+q}\}u_{n+p+q}^2
=-(-1)^{(q-i-1)/p}u_{n+i+1}^2u_{n+i+1+p}^2u_{n+q}u_{n+q+p^2}
\\
&
=(-1)^{(p+q-i-1)/p}\prod_{k=0}^{(p+q-i-1)/p}u_{n+i+1+kp}^2
\{u_{n+i-1},u_{n+p+q-1}\}|_{{\bf{u}}=\phi({\bf{u}})}.
\end{align*}
\end{itemize}
To complete this proof,  we need to prove the Jacobian identity~\eqref{C1E:JacoIden}:
\begin{equation}
\label{E:JacoIden1}
\sum_{l}\big(\Omega_{li}\frac{\partial}{\partial u_{n+l}}\Omega_{jk}+\Omega_{lj}\frac{\partial}{\partial u_{n+l}}\Omega_{ki}
+\Omega_{lk}\frac{\partial}{\partial u_{n+l}}\Omega_{ij}\big)=0,\ \forall i,j,k.
\end{equation}

Without loss of generality, we can assume that $0\leq i\leq j\leq k\leq p+q-1$. It is easy to see that the Jacobian identity holds if two indices  in $(i,j,k)$ are equal. We need to prove the  case where $0\leq i <j< k\leq p+q-1$. We also distinguish $3$ cases.
\begin{itemize}
\item Case 1 where $j-i=q$.  Since $k<p+q$, we have $0<k-j<p$. Therefore, $\Omega_{jk}=0$. Furthermore,  one can see easily that
\[
\sum_{l=0}^{p+q-1}\Omega_{lk}\frac{\partial \Omega_{ij}}{\partial u_{n+l}}=\frac{2\Omega_{ik}\Omega_{ij}}{u_{n+i}}+\frac{2\Omega_{jk}\Omega_{ij}}{u_{n+j}}=\frac{2\Omega_{ik}\Omega_{ij}}{u_{n+i}}.
\]
If $p\nmid(k-i)$, then $\Omega_{ki}=0$. Thus,   the  Jacobian identity~\eqref{E:JacoIden1} holds.
If $p|(k-i)$, then we have
\[
\sum_{l}\Omega_{lj}\frac{\partial \Omega_{ki}}{\partial u_{n+l}}\Omega_{ki}=
\sum_{i<l\leq k,\ l\equiv i\mod p}\frac{2\Omega_{lj}\Omega_{ki}}{u_{n+l}}+\frac{2\Omega_{ij}\Omega_{ki}}{u_{n+i}}=-\frac{2\Omega_{ij}\Omega_{ik}}{u_{n+i}},
\]
where we have used $\Omega_{lj}=0$ as $p\nmid (j-l)$ and $j-l\neq \pm q$ for $i<l\leq k$ and $p|(l-i)$.
\item Case 2 where $p|(j-i)$. We have
\begin{equation}
\label{E:Jaco_term_3}
\sum_{l=0}^{p+q-1}\Omega_{lk}\frac{\partial \Omega_{ij}}{\partial u_{n+l}}=
\sum_{i\leq l\leq j,\  l\equiv i\mod p}\frac{2\Omega_{lk}\Omega_{ij}}{u_{n+l}}.
\end{equation}
Since $p\leq j-i$, we have $k-j<q$.
Thus, if $k-i=q$, then $p\nmid (k-j)$. This yields
$\Omega_{jk}=0$ and
\[
\sum_{l}\Omega_{lj}\frac{\partial \Omega_{ki}}{\partial u_{n+l}}\Omega_{ki}=
\frac{2\Omega_{ij}\Omega_{ki}}{u_{n+i}}+\frac{2\Omega_{kj}\Omega_{ki}}{u_{n+k}}=\frac{2\Omega_{ij}\Omega_{ki}}{u_{n+i}}
\]
However, for $i<l\leq j$ and $l\equiv i\mod p$, we have  $\Omega_{lk}=0$. Therefore, the RHS of~\eqref{E:Jaco_term_3} becomes $2\Omega_{ik}\Omega_{ij}/u_{n+i}$. Hence,  we obtain the Jacobian identity~\eqref{E:JacoIden1}.

If $k-i\neq q$ and $p\nmid k-i$, then $\Omega_{jk}=\Omega_{ki}=\Omega_{lk}=0$ for $i<l<j$ and $l\equiv i\mod p$ (as $k-l<p+q-(i+p)\leq q$ and $p\nmid (k-l)$.
This implies that $LHS\ \eqref{E:Jaco_term_3}=0$. Thus, the Jacobian identity~\eqref{E:JacoIden1} holds.

If $p|(k-i)$, i.e. $i\equiv j\equiv k\mod p$ then the LHS of the Jacobian identity~\eqref{E:JacoIden1}
 is
 \begin{equation*}
 LHS\ \eqref{E:JacoIden1}=
\sum_{j\leq l\leq k,\  l\equiv i\mod p}\frac{2\Omega_{li}\Omega_{jk}}{u_{n+l}}.
+\sum_{i\leq l\leq k,\  l\equiv i\mod p}\frac{2\Omega_{lj}\Omega_{ki}}{u_{n+l}}.
+ \sum_{i\leq l\leq j,\  l\equiv i\mod p}\frac{2\Omega_{lk}\Omega_{ij}}{u_{n+l}}.
 \\
 \end{equation*}
 It is easy to see that the three terms in this sum can be written as follows
 \begin{align*}
\sum_{j\leq l\leq k,\  l\equiv i\mod p}\frac{2\Omega_{li}\Omega_{jk}}{u_{n+l}}
&=-\sum_{j< l\leq k,\  l\equiv i\mod p}\frac{2\Omega_{ik}\Omega_{jl}}{u_{n+l}}+\frac{2\Omega_{ji}\Omega_{jk}}{u_{n+j}},
\\
\sum_{i\leq l\leq k,\  l\equiv i\mod p}\frac{2\Omega_{lj}\Omega_{ki}}{u_{n+l}}
&
=\sum_{i\leq l<j,\  l\equiv i\mod p}\frac{2\Omega_{lj}\Omega_{ki}}{u_{n+l}}
+\sum_{j<l\leq k,\  l\equiv i\mod p}\frac{2\Omega_{lj}\Omega_{ki}}{u_{n+l}}
\\
  \sum_{i\leq l\leq j,\  l\equiv i\mod p}\frac{2\Omega_{lk}\Omega_{ij}}{u_{n+l}}
  &= \sum_{i\leq l< j,\  l\equiv i\mod p}\frac{2\Omega_{ik}\Omega_{lj}}{u_{n+l}}+\frac{2\Omega_{jk}\Omega_{ij}}{u_{n+j}}.
 \end{align*}
 Adding these three terms together we get $LHS\ \eqref{E:JacoIden1}=0$.

 \item Case 3 where $j-i\neq q$ and $p\nmid (j-i)$, i.e. we have $\Omega_{ij}=0$. Similarly, by considering the following sub-cases $k-j=q$; $p|(k-j)$; $k-j\neq q,\ \mbox{and} p\nmid (k-j)$, we can prove that the Jacobian identity~\eqref{E:JacoIden1} holds.
\end{itemize}
\end{proof}

Analogously, we can derive the following result for the map obtained from the plus-KdV.
\begin{theorem}The  map obtained as the $(q,-p)$ reduction of equation~\eqref{E:plusKdV} admits the following  Poisson structure for $0\leq i <j\leq p+q-1$.
\begin{itemize}
\item If $p=1$,  we have
\begin{equation}
\label{E:poi_plus_KdV1}
\{u_{n+i},u_{n+j}\}=
\left\{
\begin{array}{ll}
(-1)^{j-i}\prod_{r=i}^j u_{n+i}^2, & 0<j-i<p+q-1,\\
\big(-1+(-1)^{p+q-1}\prod_{r=i+1}^{j-1}u_{n+r}^2\big)u_{n+i}^2u_{n+j}^2, & j-i=p+q-1.
\end{array}
\right.
\end{equation}
\item If $p>1$, we have
\end{itemize}
\begin{equation}
\label{E:poi_plus_KdV}
\{u_{n+i},u_{n+j}\}=
\left\{
\begin{array}{ll}
-u_{n+i}^2u_{n+j}^2& \mbox{if}\  j-i=q,\\
(-1)^{(j-i)/p}\prod_{k=0}^{(j-i)/p}u_{n+i+kp}^2,& \mbox{if} \ j-i\equiv 0 \mod p,\\
0&\mbox{otherwise}.
\end{array}
\right.
\end{equation}
\end{theorem}
\subsection{Examples}
The matrices
\[
L=
\begin{pmatrix}
u_{l,m}-\frac{1}{u_{l+1,m}} & \lambda \\
1 & 0
\end{pmatrix},\qquad
M=
\begin{pmatrix}
u_{l,m} & \lambda \\
1 & \frac{1}{u_{l,m}},
\end{pmatrix}
\]
are Lax-matrices for the KdV-equation (\ref{E:KdV}). For the (3,-2)-reduction (with $n=2l+3m$, $u_{l,m}\to v_n$, $L_{l,m}\to L_n$, $M_{l,m}\to M_n$),
\[
v_{n+2}-v_{n+3}=\frac{1}{v_n}-\frac{1}{v_{n+5}},
\]
the trace of the monodromy matrix
$
\mathcal{L}=M_0^{-1} L_1M_1^{-1} L_2 L_0
$
yields the following three invariants
\begin{align*}
I_1=\frac{1}{v_0}+\frac{1}{v_1}+\frac{1}{v_2}+\frac{1}{v_3}+\frac{1}{v_4}-v_2,\quad
I_2=\frac{(v_2v_4-1)(v_1v_3-1)(v_0v_2-1)}{v_0v_1v_2v_3v_4},\quad
I_3=\frac{(v_1v_4+1)(v_0v_3+1)}{v_0v_1v_2v_3v_4}.
\end{align*}
We have
\[
\Omega=\begin{small}
\begin{pmatrix}
0 & 0 & -v_0^2v_2^2 & v_0^2v_3^2 & v_0^2v_2^2v_4^2 \\
0 & 0 & 0 & -v_1^2v_3^2 & v_1^2v_4^2 \\
v_0^2v_2^2  & 0 & 0 & 0 & -v_2^2v_4^2 \\
-v_0^2v_3^2  & v_1^2v_3^2 & 0 & 0 & 0 \\
-v_0^2v_2^2v_4^2  & -v_1^2v_4^2 & v_2^2v_4^2 & 0 & 0
\end{pmatrix}\end{small},
\]
and $I_1$ is a Casimir of the corresponding bracket. The rank of the bracket is
4, and we have $5-4/2=3$ functionally independent invariants in involution with respect to the Poisson bracket, so the Poisson map is completely integrable.

The matrices
\[
L=
\begin{pmatrix}
u_{l,m}-\frac{1}{u_{l+1,m}} & \lambda \\
1 & 0
\end{pmatrix},\qquad
M=
\begin{pmatrix}
u_{l,m} & \lambda \\
-1 & -\frac{1}{u_{l,m}},
\end{pmatrix}
\]
are anti-Lax-matrices for the plus-KdV-equation (\ref{E:plusKdV}), which satisfy
$L_{l,m+1}M_{l,m}\equiv -M_{l+1,m}L_{l,m}$ modulo the equation. The idea of an anti-Lax pair was suggested to one of the authors by J.A.G. Roberts \cite{JB}. For the (3,-2)-reduction,
\[
v_{n+2}+v_{n+3}=\frac{1}{v_n}+\frac{1}{v_{n+5}},
\]
the trace of the monodromy matrix
gives rise to the following three 2-integrals
\begin{align*}
J_1=\frac{1}{v_0}-\frac{1}{v_1}+\frac{1}{v_2}-\frac{1}{v_3}+\frac{1}{v_4}-v_2,\quad
J_2=\frac{(v_2v_4-1)(v_1v_3-1)(v_0v_2-1)}{v_0v_1v_2v_3v_4},\quad
J_3=\frac{(v_1v_4-1)(v_0v_3-1)}{v_0v_1v_2v_3v_4},
\end{align*}
that is, we have $J_i(v_1,\ldots,v_5)\equiv -J_i(v_0,\ldots,v_4)$.
We have
\[
\Omega=\begin{small}
\begin{pmatrix}
0 & 0 & -v_0^2v_2^2 & -v_0^2v_3^2 & v_0^2v_2^2v_4^2 \\
0 & 0 & 0 & -v_1^2v_3^2 & -v_1^2v_4^2 \\
v_0^2v_2^2  & 0 & 0 & 0 & -v_2^2v_4^2 \\
v_0^2v_3^2  & v_1^2v_3^2 & 0 & 0 & 0 \\
-v_0^2v_2^2v_4^2  & v_1^2v_4^2 & v_2^2v_4^2 & 0 & 0
\end{pmatrix}\end{small},
\]
and $J_1$ is a Casimir of the corresponding bracket. The rank of the bracket is
4, and we have $5-4/2=3$ functionally independent invariants in involution, namely
$I_1=J_1/J_2$, $I_2=J_2/J_3$, and $I_3=J_2J_3$, so the Poisson map is completely integrable.

\section{Poisson brackets for $(q,-p)$ reductions of the Lotka-Volterra equation, and the Liouville equation}
 The discrete Lotka-Volterra equation is given as follows
 \begin{equation}
 \label{E:Lk}
 u_{l,m+1}(1+u_{l,m})=u_{l+1,m}(1+u_{l+1,m+1}).
 \end{equation}
 We introduce the new variable $w_{l,m}=\ln u_{l,m}$ to get
 \begin{equation}
 \label{E:LK_log}
w_{l,m+1}-w_{l+1,m}=-\ln (1+w_{l,m})+\ln (1+w_{l+1,m+1}).
 \end{equation}
This equation's functional form is similar to that of the KdV equation. Therefore, it suggests that one could try to introduce $w_{l,m}=v_{l-1,m-1}- v_{l,m}$ to
write equation~\eqref{E:LK_log} as an Euler-Lagrange equation. One can easily derive a Lagrangian
\begin{equation}
\label{L:lagrangian_LV}
L= v_{l,m}v_{l,m+1}-v_{l,m}v_{l+1,m}+F(v_{l,m}-v_{l+1,m+1}),
\end{equation}
 where
 $$
 F(x)=\int_0^x\ln(1+e^t)dt.
$$
Using the discrete analogue of the Ostragrasky transformation, we can derive Poisson
brackets of mappings obtained  the $(p,-q)$ reduction of equation~\eqref{E:LK_log} in the $v$ variable and then in the $w$ variable.
Once we have a Poisson bracket in the $w$ variable, the associate Poisson bracket in the original variable is calculated by
\[
\{u_{n+i},u_{n+j}\}=\{e^{w_{n+i}}, e^{w_{n+j}}\}=e^{w_{n+i}}e^{w_{n+j}}\{w_{n+i},w_{n+j}\}.
\]
We obtain the following  theorem.
\begin{theorem}
The  map obtained as the $(q,-p)$ reduction of equation~\eqref{E:Lk} admits the following  Poisson structure for $0\leq i <j\leq p+q-1$.
\begin{itemize}
\item If $p=1$,  we have
\begin{equation}
\label{E:poi_Lk1}
\{u_{n+i},u_{n+j}\}=
\left\{
\begin{array}{ll}
(1+u_{n+i})(1+u_{n+j})\prod_{r=i+1}^{j-1}\big(\frac{1+u_{n+r}}{u_{n+r}}\big), & 0<j-i<p+q-1,\\
(1+u_{n+i})(1+u_{n+j})\Big(\prod_{r=i+1}^{j-1}\big(\frac{1+u_{n+r}}{u_{n+r}}\big)-1\Big), & j-i=p+q-1.
\end{array}
\right.
\end{equation}
\item If $p>1$, we have
\end{itemize}
\begin{equation}
\label{E:poi_Lk}
\{u_{n+i},u_{n+j}\}=
\left\{
\begin{array}{ll}
-(u_{n+i}+1)(u_{n+j}+1)& \mbox{if}\  j-i=q,\\
(u_{n+i}+1)(u_{n+j}+1)\prod_{k=1}^{(j-i)/p-1}\Big(\frac{u_{n+i+kp}+1}{u_{n+i+kp}}\Big),& \mbox{if} \ j-i\equiv 0 \mod p,\\
0&\mbox{otherwise}.
\end{array}
\right.
\end{equation}
\end{theorem}
Similarly, if we take
$$L= v_{l,m}v_{l,m+1}+v_{l,m}v_{l+1,m}+F(v_{l,m}+v_{l+1,m+1}),$$
and $w_{l,m}=v_{l-1,m-1}+ v_{l,m}$ we obtain Poisson bracket for  the $(q,-p)$ reduction of the following equation
\begin{equation}
\label{E:plus_Lk}
(1+u_{l,m})(1+u_{l+1,m+1})u_{l+1,m}u_{l,m+1}-1=0.
\end{equation}
%************************************************************************************************
\begin{theorem}
The  map obtained as the $(q,-p)$ reduction of equation~\eqref{E:plus_Lk} admits the following  Poisson structure for $0\leq i <j\leq p+q-1$.
\begin{itemize}
\item If $p=1$,  we have
\begin{equation}
\label{E:poi_plus_Lk1}
\{u_{n+i},u_{n+j}\}=
\left\{
\begin{array}{ll}
(-1)^{j-i}(1+u_{n+i})(1+u_{n+j})\prod_{r=i+1}^{j-1}\big(\frac{1+u_{n+r}}{u_{n+r}}\big), & 0<j-i<p+q-1,\\
(1+u_{n+i})(1+u_{n+j})\Big((-1)^{j-i}\prod_{r=i+1}^{j-1}\big(\frac{1+u_{n+r}}{u_{n+r}}\big)-1\Big), & j-i=p+q-1.
\end{array}
\right.
\end{equation}
\item If $p>1$, we have
\end{itemize}
\begin{equation}
\label{E:poi_plus_Lk}
\{u_{n+i},u_{n+j}\}=
\left\{
\begin{array}{ll}
-(u_{n+i}+1)(u_{n+j}+1)& \mbox{if}\  j-i=q,\\
(-1)^{(j-i)/p}(u_{n+i}+1)(u_{n+j}+1)\prod_{k=1}^{(j-i)/p-1}\Big(\frac{u_{n+i+kp}+1}{u_{n+i+kp}}\Big),& \mbox{if} \ j-i\equiv 0 \mod p,\\
0&\mbox{otherwise}.
\end{array}
\right.
\end{equation}
\end{theorem}
%***************************************************************************************
%\begin{remark}
If we take $$L= v_{l,m}v_{l,m+1}+v_{l,m}v_{l+1,m}-F(v_{l,m}+v_{l+1,m+1}),$$
we obtain Poisson brackets of maps obtained as reductions of the discrete Liouville equation
\begin{equation}
\label{E:Liouville}
(1+u_{l,m})(1+u_{l+1,m+1})-u_{l+1,m}u_{l,m+1}=0.
\end{equation}
\begin{theorem}
The  map obtained as the $(q,-p)$ reduction of (\ref{E:Liouville}) admits the following  Poisson structure for $0\leq i <j\leq p+q-1$.
\begin{itemize}
\item If $p=1$,  we have
\begin{equation}
\label{E:poi_Louville1}
\{u_{n+i},u_{n+j}\}=
\left\{
\begin{array}{ll}
-(1+u_{n+i})(1+u_{n+j})\prod_{r=i+1}^{j-1}\big(\frac{1+u_{n+r}}{u_{n+r}}\big), & 0<j-i<p+q-1,\\
-(1+u_{n+i})(1+u_{n+j})\Big(\prod_{r=i+1}^{j-1}\big(\frac{1+u_{n+r}}{u_{n+r}}\big)+1\Big), & j-i=p+q-1.
\end{array}
\right.
\end{equation}
\item If $p>1$, we have
\end{itemize}
\begin{equation}
\label{E:poi_Louville}
\{u_{n+i},u_{n+j}\}=
\left\{
\begin{array}{ll}
-(u_{n+i}+1)(u_{n+j}+1)& \mbox{if}\  j-i=q,\\
-(u_{n+i}+1)(u_{n+j}+1)\prod_{k=1}^{(j-i)/p-1}\Big(\frac{u_{n+i+kp}+1}{u_{n+i+kp}}\Big),& \mbox{if} \ j-i\equiv 0 \mod p,\\
0&\mbox{otherwise}.
\end{array}
\right.
\end{equation}
\end{theorem}

\begin{remark}
We note that all equations that we have considered in this paper can be written in the following form (after some transformations)
$$u_{l+1,m}-u_{l,m+1}=G(u_{l,m})-G(u_{l+1,m+1}) \ \mbox{or}\
u_{l+1,m}+u_{l,m+1}=G(u_{l,m})+G(u_{l+1,m+1}).
$$
In fact, using the Ostrogradsky transformation and a Lagrangian which has the same form as \eqref{L:lagrangian_LV}, one can always  find Poisson brackets for maps obtained as reductions of these equations. This method can also be applied to more general $(q,p)$-reductions, i.e. we do not require $\rm gcd(p,q)=1$.
\end{remark}

\subsection{Examples}
The matrices
\[
L=
\begin{pmatrix}
1 & -\frac{u_{l+1,m}}{\lambda^2} \\
\frac{1}{u_{l+1,m}+1} & \frac{u_{l+1,m}}{u_{l+1,m}+1}
\end{pmatrix},\qquad
M=
\begin{pmatrix}
1 & -\frac{u_{l,m+1}(u_{l+1,m}+1)}{\lambda^2} \\
1 & 0,
\end{pmatrix}
\]
are Lax-matrices for the Lotka-Volterra-equation (\ref{E:Lk}). For the (3,-2)-reduction,
\[
v_{n+3}(v_{n}+1)=v_{n+2}(v_{n+5}+1),
\]
the determinant of the monodromy matrix
provides the invariant
\[
I_1=\frac{(v_0+1)(v_1+1)(v_2+1)(v_3+1)(v_4+1)}{v_2}
\]
which is also a Casimir of the Poisson bracket defined by
\[
\Omega=\begin{tiny}
\begin{pmatrix}
0 & 0 &(v_0+1)(v_2+1) & -(v_0+1)(v_3+1) & (v_0+1)(\frac{1}{v_2}+1)(v_4+1)\\
0 & 0 & 0 & (v_1+1)(v_3+1) & -(v_1+1)(v_4+1) \\
-(v_0+1)(v_2+1) & 0 & 0 & 0 & (v_2+1)(v_4+1) \\
(v_0+1)(v_3+1) & -(v_1+1)(v_3+1) & 0 & 0 & 0 \\
-(v_0+1)(\frac{1}{v_2}+1)(v_4+1) &  (v_1+1)(v_4+1)& -(v_2+1)(v_4+1) &  0 & 0
\end{pmatrix}\end{tiny}.
\]

The trace yields two additional functionally independent integrals,
\[
I_2=v_0v_3-v_2+v_1v_4,\quad
I_3=(v_1+1)(v_0+v_2+2)+(v_3+1)(v_2+v_4+2)+\frac{(v_0+1)(v_1+v_3+1)(v_4+1)}{v_2}.
\]
These integrals are in involution which shows the mapping is completely integrable.

The (3,-2)-reduction of equation (\ref{E:plus_Lk}),
\[
(1+v_{n})v_{n+2}v_{n+3}(1+v_{n+5})=1,
\]
is a Poisson map with Poisson structure
\[
\Omega=\begin{tiny}
\begin{pmatrix}
0 & 0 &-(v_0+1)(v_2+1) & -(v_0+1)(v_3+1) & (v_0+1)(\frac{1}{v_2}+1)(v_4+1)\\
0 & 0 & 0 & -(v_1+1)(v_3+1) & -(v_1+1)(v_4+1) \\
(v_0+1)(v_2+1) & 0 & 0 & 0 & -(v_2+1)(v_4+1) \\
(v_0+1)(v_3+1) & (v_1+1)(v_3+1) & 0 & 0 & 0 \\
-(v_0+1)(\frac{1}{v_2}+1)(v_4+1) &  (v_1+1)(v_4+1)& (v_2+1)(v_4+1) &  0 & 0
\end{pmatrix}\end{tiny}.
\]
The Casimir function
\[
C=\frac{(v_0+1)v_2(v_2+1)(v_4+1)}{(v_1+1)(v_3+1)}
\]
is a 2-integral, applying the map gives the multiplicative inverse. Hence the
function $C+1/C$ is an integral. We have not been able to find more integrals and, as the map seems to have non-vanishing entropy, we believe it is not integrable.

The (3,-2)-reduction of equation (\ref{E:Liouville}),
\begin{equation}\label{Lior}
(1+v_{n})(1+v_{n+5})=v_{n+2}v_{n+3},
\end{equation}
is a Poisson map with Poisson structure
\[
\Omega=\begin{tiny}
\begin{pmatrix}
0 & 0 &-(v_0+1)(v_2+1) & -(v_0+1)(v_3+1) & -(v_0+1)(\frac{1}{v_2}+1)(v_4+1)\\
0 & 0 & 0 & -(v_1+1)(v_3+1) & -(v_1+1)(v_4+1) \\
(v_0+1)(v_2+1) & 0 & 0 & 0 & -(v_2+1)(v_4+1) \\
(v_0+1)(v_3+1) & (v_1+1)(v_3+1) & 0 & 0 & 0 \\
(v_0+1)(\frac{1}{v_2}+1)(v_4+1) &  (v_1+1)(v_4+1)& (v_2+1)(v_4+1) &  0 & 0
\end{pmatrix}\end{tiny}.
\]
The Casimir function
\[
C=\frac{(v_0+1)(v_2+1)(v_4+1)}{(v_1+1)v_2(v_3+1)}
\]
is a 2-integral with $C^\prime=1/C$ and so provides us with an integral $I_1=C+1/C$.
Other $k$-integrals can be found from the $i$- and $j$-integrals given in \cite{DiscreteLiouville}
for an equation, \cite[Eq. (19)]{DiscreteLiouville}, which is related to (\ref{E:Liouville}) by
$u_{i,j}\to\frac{-1}{1+u_{l,m}}$. The reduction of these integrals yield the 2-integral
\[
A=\frac{(v_0+v_3+1)(v_1+v_4+1)}{(v_1+1)(v_3+1)}
\]
and the 3-integral
\[
B=\frac{(v_0+v_2+1)(v_1+v_3+1)}{(v_0+1)(v_3+1)}.
\]
Applying the map $v_i\mapsto v_{i+1}$ subject to (\ref{Lior}) we obtain
\[
A^\prime=\frac{(v_0+v_3+1)v_2(v_1+v_4+1)}{(v_0+1)(v_2+1)(v_4+1)},\quad A^{\prime\prime}=A
\]
and
\[
B^\prime=\frac{(v_1+v_3+1)(v_2+v_4+1)}{(v_1+1)(v_4+1)},\quad
B^{\prime\prime}=\frac{(v_0+v_2+1)(v_2+v_4+1)}{v_2(v_2+1)},\quad
B^{\prime\prime\prime}=B.
\]
We have $C=A/A^\prime$ and it can be checked that
the integrals $I_1$, $I_2=A+A^\prime$, $I_3=BB^\prime B^{\prime\prime}$ are functionally independent and in involution. The additional integral $I_4=B+B^\prime+B^{\prime\prime}$ is functionally independent, and in involution with $I_2$ (and $I_1$), but not with $I_3$. A linearisation of the lattice equation can be found in \cite{DiscreteLiouville}.

\section*{Acknowledgment}
This research was supported by the Australian Research Council and by La Trobe University's Disciplinary Research Program in Mathematical and Computing Sciences.

\end{document}